\documentclass[conference]{IEEEtran}

\usepackage[cmex10]{amsmath}
\usepackage[dvipdfmx]{graphicx}
\usepackage{array}
\usepackage{mdwmath}
\usepackage{mdwtab}
\usepackage{amssymb}

\newtheorem{theorem}{Theorem}
\newtheorem{lemma}{Lemma}
\newtheorem{definition}{Definition}
\newtheorem{corollary}{Corollary}
\newtheorem{remark}{Remark}

\newcommand{\qed}{\hfill \IEEEQED}

\DeclareMathAlphabet{\bm}{OML}{cmm}{b}{it}

\newcommand{\bol}[1]{\mathbf{#1}}
\newcommand{\rom}[1]{\mathrm{#1}}

\allowdisplaybreaks[2]

\renewcommand{\baselinestretch}{0.95}
\renewcommand{\abovedisplayskip}{1mm}
\renewcommand{\belowdisplayskip}{1mm}
\renewcommand{\abovedisplayshortskip}{1mm}
\renewcommand{\belowdisplayshortskip}{1mm}

\begin{document}
%
\title{Non-Asymptotic Analysis of Privacy Amplification via R\'enyi Entropy and Inf-Spectral Entropy}



%
\author{\IEEEauthorblockN{Shun Watanabe\IEEEauthorrefmark{1} and
Masahito Hayashi\IEEEauthorrefmark{2}}
\IEEEauthorblockA{\IEEEauthorrefmark{1}
Department of  Information Science and Intelligent Systems, 
University of Tokushima, Tokushima, Japan, \\
Email: shun-wata@is.tokushima-u.ac.jp \\
\IEEEauthorrefmark{2}Graduate School of Mathematics, Nagoya University, Japan, \\
and Centre for Quantum Technologies, National University of Singapore, Singapore. \\
E-mail: masahito@math.nagoya-u.ac.jp}}


\maketitle

\begin{abstract}
This paper investigates the privacy amplification problem, 
and compares the existing two bounds: the exponential bound
derived by one of the authors and the min-entropy bound derived
by Renner. It turns out that the exponential bound is better than 
the min-entropy bound when a security parameter is rather small
for a block length, and that the min-entropy bound is better than 
the exponential bound when a security parameter is rather large
for a block length. 
Furthermore, we present another bound that interpolates
the exponential bound and the min-entropy bound by a hybrid use of
the R\'enyi entropy and the inf-spectral entropy.
\end{abstract}


%
\IEEEpeerreviewmaketitle

\section{Introduction}

The privacy amplification \cite{bennett:95} is a technique to distill a secret key
from a source that is partially known to an eavesdropper, usually referred to
as Eve. The privacy amplification is regarded as an indispensable tool
in the information theoretic security, and it has been studied in many
literatures (eg. \cite{renner:05b,renner:05d,tomamichel:phd,tomamichel:11b,hayashi:10,hayashi:10b,tomamichel:12}). 

Recently, the non-asymptotic analysis of coding problems has attracted 
considerable attention  \cite{hayashi:09, polyanskiy:10}.
Especially for the channel coding problem, the relation between
various types of non-asymptotic bounds are extensively compared in
\cite{polyanskiy:10}.  

The performance of the privacy amplification is typically
characterized by the smooth minimum entropy or the 
inf-spectral entropy \cite{renner:05b,datta:09,tomamichel:12}. 
There is also another approach, the exponential bound, which 
has been investigated by one of the authors \cite{hayashi:10b}.
So far, the relation between non-asymptotic bounds derived by these two approaches
has not been clarified.
The first purpose of this paper is to compare
the min-entropy bound, which is derived by the smooth minimum entropy framework,
and the exponential bound. Actually, it turns out that the exponential bound is better than
the min-entropy bound when a security parameter 
is rather small for a block length. In the following, we explain 
a reason for this result.

In the achievability part of 
the smooth entropy framework \cite{renner:05b} 
or the information spectrum approach \cite{han:book},
a performance criterion of a problem, such as an error probability or a security 
parameter, is usually upper bounded by a formula consisting of two terms.
One of the terms is caused by the smoothing error, which corresponds to
a tail probability of atypical outcomes. The other is caused by typical outcomes.
In the following, let us call the former one the type 1 error term and the latter one
the type 2 error term respectively.

To derive a tight bound in total, we need to tightly bound
both the type 1 and type 2 error terms.
This fact has been recognized in literatures.
Indeed, one of the authors derived the state-of-the-art error exponent
for the classical-quantum channel coding by tightly bounding both types of
error terms \cite{hayashi:07b}.
In \cite{polyanskiy:10}, Polyanskiy {\em et.~al.}~derived a non-asymptotic
bound of the channel coding, which is called the DT bound, by tightly bounding both types of
error terms. The DT bound remarkably improves on the so-called
Feinstein bound because the type 2 error term is loosely bounded
in the Feinstein bound. The improvement is especially remarkable
when a required error probability is rather small for a block length.

For the privacy amplification problem, 
one of the authors derived the state-of-the-art exponent
of the variational distance by tightly bounding both types of
error terms \cite{hayashi:10b}. 
On the other hand, the type 2 error term is loosely bounded in
the bound derived via the smooth minimum entropy \cite{renner:05b}.

As is expected from the above argument,
the exponential bound turns out to be better
than the min-entropy bound when a security parameter 
is rather small for a block length.
For rather large security parameters, the min-entropy bound is better
than the exponential bound. This is because we derive the exponential bound
by using the large deviation technique \cite{dembo-zeitouni-book}.
The large deviation technique is only tight when a threshold of  a tail probability
is away from the average, and this is not the case when a security parameter
is rather large for a block length.
As the second purpose of this paper, we derive a bound that 
interpolates the exponential bound and the min-entropy bound.
This is done by a hybrid use of the R\'enyi entropy and the inf-spectral entropy.
It turns out that the hybrid bound is better than both the exponential bound
and the min-entropy bound for whole ranges of security parameters.

The rest of the paper is organized as follows.
In Section \ref{section:preliminaries}, we summarize known bounds on
the privacy amplification. In Section \ref{section:hybrid}, we propose
a novel bound by using the R\'enyi entropy and the inf-spectral entropy.
In Section \ref{section:numerical-calcuation}, we compare the bounds numerically.

\section{Preliminaries}
\label{section:preliminaries}

In this section, we review the problem setting and known results on
the privacy amplification. Although most of results in this section were stated
explicitly or implicitly in literatures, we restate them for reader's convenience.
Especially, Theorem \ref{theorem:one-shot}, 
Theorem \ref{theorem:spectrum}, and Theorem \ref{theorem:gaussian-approximation} are
classical analogue of those obtained in \cite{tomamichel:12} for the quantum setting,
where the distance to evaluate the smoothing is different. 

\subsection{Problem Formulation}

For a set ${\cal A}$, let ${\cal P}({\cal A})$ be the set of all 
probability distribution on ${\cal A}$. It is also convenient to introduce
the set $\bar{{\cal P}}({\cal A})$ of all sub-normalized non-negative functions.

Let $P_{XZ} \in \bar{{\cal P}}({\cal X} \times {\cal Z})$ be a sub-normalized non-negative function.
For a function $f: {\cal X} \to {\cal S}$ and the key $S = f(X)$, let
\begin{eqnarray*}
P_{SZ}(s,z) = \sum_{x \in f^{-1}(s)} P_{XZ}(x,z).
\end{eqnarray*}
We define the security by
\begin{eqnarray*}
d(f|P_{XZ}) = d(P_{SZ}, P_{\bar{S}} \times P_Z),
\end{eqnarray*}
where $P_{\bar{S}}$ is the uniform distribution on ${\cal S}$
and 
\begin{eqnarray}
\label{eq:definition-of-distance}
d(P,Q) := \frac{1}{2} \sum_a | P(a) - Q(a)|
\end{eqnarray}
for $P,Q \in \bar{{\cal P}}({\cal A})$.

Although the quantity $d(f|P_{XZ})$ has no operational meaning 
for unnormalized $P_{XZ}$, it will be used to derive bounds on $d(f|P_{XZ})$
for normalized $P_{XZ}$. For distribution $P_{XZ} \in {\cal P}({\cal X} \times {\cal Z})$ and 
security parameter $\varepsilon \ge 0$, we are interested in characterizing 
\begin{eqnarray*}
\ell(P_{XZ},\varepsilon) = \sup\{ \log |{\cal S}| : \exists f: {\cal X} \to {\cal S} \mbox{ s.t. } d(f|P_{XZ}) \le \varepsilon \}.
\end{eqnarray*}

\subsection{Min Entropy Framework}
\label{subsection:min-entropy}

In this section, we review the smooth minimum entropy framework
that was mainly introduced and developed by Renner and 
his collaborators \cite{renner:05b,renner:05d,tomamichel:09,tomamichel:10,tomamichel:phd}.

\begin{definition}
For $P_{XZ} \in \bar{{\cal P}}({\cal X} \times {\cal Z})$ and a normalized $R_Z \in {\cal P}({\cal Z})$,
we define
\begin{eqnarray*}
H_{\min}(P_{XZ}|R_Z) = - \log \max_{x \in {\cal X} \atop z \in \rom{supp}(R_Z)} \frac{P_{XZ}(x,z)}{R_Z(z)}.
\end{eqnarray*}
Then, we define
\begin{eqnarray*}
\bar{H}_{\min}^\varepsilon(P_{XZ}|R_Z) = \max_{Q_{XZ} \in \bar{{\cal B}}^\varepsilon(P_{XZ})} H_{\min}(Q_{XZ}|R_Z),
\end{eqnarray*}
where 
\begin{eqnarray*}
\bar{{\cal B}}^\varepsilon(P_{XZ}) = \left\{ Q_{XZ} \in \bar{{\cal P}}({\cal X} \times {\cal Z}) : d(P_{XZ}, Q_{XZ}) \le \varepsilon \right\}.
\end{eqnarray*}
We also define
\begin{eqnarray*}
H_{\min}^\varepsilon(P_{XZ}|R_Z) = \max_{Q_{XZ} \in {\cal B}^\varepsilon(P_{XZ})} H_{\min}(Q_{XZ}|R_Z),
\end{eqnarray*}
where 
\begin{eqnarray*}
{\cal B}^\varepsilon(P_{XZ}) = \left\{ Q_{XZ} \in {\cal P}({\cal X} \times {\cal Z}) : d(P_{XZ}, Q_{XZ}) \le \varepsilon \right\}.
\end{eqnarray*}
\end{definition}

The following is a key lemma to derive 
every lower bound on $\ell(P_{XZ},\varepsilon)$.
\begin{lemma}[Leftover Hash:\cite{renner:05b}]
\label{lemma:left-over}
Let $F$ be the uniform random variable on a set of universal 2 hash family ${\cal F}$. Then, 
for $P_{XZ} \in \bar{{\cal P}}({\cal X} \times {\cal Z})$ and
$R_Z \in {\cal P}({\cal Z})$, we have\footnote{Technically,
$R_Z$ must be such that $\rom{supp}(P_Z) \subset \rom{supp}(R_Z)$.}
\begin{eqnarray*}
\mathbb{E}_F[d(F|P_{XZ}) ] \le \frac{1}{2} \sqrt{|{\cal S}| e^{- H_2(P_{XZ}|R_Z)}},
\end{eqnarray*}
where 
\begin{eqnarray*}
H_2(P_{XZ}|R_Z) = - \log \sum_{x \in {\cal X} \atop z \in \rom{supp}(R_Z)} \frac{P_{XZ}(x,z)^2}{R_Z(z)}
\end{eqnarray*}
is the conditional R\'enyi entropy of order $2$ relative to $R_Z$.
\end{lemma}

Since $H_2(P_{XZ}|R_Z) \ge H_{\min}(P_{XZ}|R_Z)$, we have the following.
\begin{corollary}
For $P_{XZ} \in \bar{{\cal P}}({\cal X} \times {\cal Z})$ and
$R_Z \in {\cal P}({\cal Z})$, we have
\begin{eqnarray*}
\mathbb{E}_F[d(F|P_{XZ}) ] \le \frac{1}{2} \sqrt{|{\cal S}| e^{- H_{\min}(P_{XZ}|R_Z)}}.
\end{eqnarray*}
\end{corollary}

Furthermore, since
\begin{eqnarray*}
d(P_{XZ}|f) \le 2 \varepsilon + d(\bar{P}_{XZ}|f)
\end{eqnarray*}
holds for $\bar{P}_{XZ} \in \bar{{\cal B}}^\varepsilon(P_{XZ})$ by the triangular inequality, we have the following.
\begin{corollary}
\label{corollary:smooth-entropy-bound}
For $P_{XZ} \in {\cal P}({\cal X} \times {\cal Z})$ and $R_Z \in {\cal P}({\cal Z})$, we have
\begin{eqnarray*}
\mathbb{E}_F[d(F|P_{XZ}) ] \le 2 \varepsilon + \frac{1}{2} \sqrt{|{\cal S}| e^{- \bar{H}_{\min}^\varepsilon(P_{XZ}|R_Z)}}.
\end{eqnarray*}
\end{corollary}

The following is a key lemma to derive 
a upper bound on $\ell(P_{XZ},\varepsilon)$.
\begin{lemma}[Monotonicity]
\label{lemma:monotonicity}
For any function $f:{\cal X} \to {\cal S}$,
$P_{XZ} \in {\cal P}({\cal X} \times {\cal Z})$, and $R_Z \in {\cal P}({\cal Z})$, we have
\begin{eqnarray*}
H_{\min}^\varepsilon(P_{SZ}|R_Z) \le H_{\min}^\varepsilon(P_{XZ}|R_Z).
\end{eqnarray*}
\end{lemma}
\begin{proof}
See Appendix \ref{appendix:lemma:monotonicity}.
\end{proof}

\begin{remark}
When Eve's side-information is the quantum density operator instead of
the random variable, the monotonicity of the smooth minimum entropy was
proved in \cite[Proposition 3]{tomamichel:12}, where the smoothing
is evaluated by the so-called purified distance instead of the trace distance.
For the quantum setting and the trace distance, it is not clear whether the monotonicity holds
or not because we cannot apply Uhlmann's theorem to the trace distance directly.
\end{remark}

From Corollary \ref{corollary:smooth-entropy-bound} and Lemma \ref{lemma:monotonicity},
we get the following lower and upper bounds on $\ell(P_{XZ},\varepsilon)$.
\begin{theorem}
\label{theorem:one-shot}
For any $0 < \eta \le \varepsilon$, we have
\begin{eqnarray*}
&& \hspace{-10mm} \max_{R_Z \in {\cal P}({\cal Z})} \bar{H}_{\min}^{(\varepsilon - \eta)/2}(P_{XZ}|R_Z) + \log 4\eta^2 - 1  \\
&\le& \ell(P_{XZ},\varepsilon) \\
&\le& H_{\min}^\varepsilon(P_{XZ}|P_Z).
\end{eqnarray*}
\end{theorem}

\subsection{Information Spectrum Approach}
\label{subsection:information-spectrum}

In this section, we introduce the inf-spectral  entropy.
The quantity is used to calculate the lower and upper bounds 
in Theorem \ref{theorem:one-shot}.

\begin{definition}
For $P_{XZ} \in {\cal P}({\cal X} \times {\cal Z})$ and  $0 \le \varepsilon \le 1$, let 
\begin{eqnarray*}
\lefteqn{ H_\rom{s}^\varepsilon(P_{XZ}|R_Z) } \\
&:=& \sup\left\{ r : P_{XZ}\left\{ - \log \frac{P_{XZ}(x,z)}{R_Z(z)} \le r \right\} \le \varepsilon \right\}
\end{eqnarray*}
be the conditional inf-spectral entropy relative to $R_Z \in {\cal P}({\cal Z})$.
\end{definition}

The following two lemmas relate the quantities 
$H_{\min}^\varepsilon(P_{XZ}|R_Z)$ and $H_\rom{s}^{\varepsilon}(P_{XZ}|R_Z)$.
\begin{lemma}
\label{lemma:spectrum-direct-bound}
For $P_{XZ} \in {\cal P}({\cal X} \times {\cal Z})$ and $R_Z \in {\cal P}({\cal Z})$, we have
\begin{eqnarray*}
\bar{H}_{\min}^{\varepsilon/2}(P_{XZ}|R_Z) \ge H_\rom{s}^\varepsilon(P_{XZ}|R_Z).
\end{eqnarray*}
\end{lemma}
\begin{proof}
See Appendix \ref{proof:lemma:spectrum-direct-bound}.
\end{proof}

\begin{lemma}
\label{lemma:spectrum-converse-bound}
For $P_{XZ} \in {\cal P}({\cal X} \times {\cal Z})$, we have
\begin{eqnarray*}
H_{\min}^{\varepsilon}(P_{XZ}|P_Z) \le H_\rom{s}^{\varepsilon+\zeta}(P_{XZ}|P_Z) - \log \zeta
\end{eqnarray*}
for any $0 < \zeta \le 1 - \varepsilon$.
\end{lemma}
\begin{proof}
See Appendix \ref{proof:lemma:spectrum-converse-bound}.
\end{proof}

From Theorem \ref{theorem:one-shot}, 
Lemma \ref{lemma:spectrum-direct-bound} and 
Lemma \ref{lemma:spectrum-converse-bound}, we have the following.
\begin{theorem}
\label{theorem:spectrum}
For any $0 < \eta \le \varepsilon$ and $0 < \zeta \le 1 - \varepsilon$, we have
\begin{eqnarray*}
\lefteqn{ \max_{R_Z \in {\cal P}({\cal Z})} H_\rom{s}^{\varepsilon - \eta}(P_{XZ}|R_Z) + \log 4 \eta^2 - 1 } \\
&\le& \ell(P_{XZ},\varepsilon) \\
&\le& H_\rom{s}^{\varepsilon + \zeta}(P_{XZ}|P_Z) - \log \zeta.
\end{eqnarray*}
\end{theorem}

\subsection{Gaussian Approximation}
\label{subsection:gaussian}

In this section, we consider the asymptotic setting.
By applying the Berry-Ess\'een theorem to Theorem \ref{theorem:spectrum},
we have the following Gaussian approximation of $\ell(P^n_{XZ},\varepsilon)$.
\begin{theorem}
\label{theorem:gaussian-approximation}
Let 
\begin{eqnarray*}
V(X|Z) := \sum_{x,z} P_{XZ}(x,z)\left( \log \frac{1}{P_{X|Z}(x|z)} - H(X|Z) \right)^2
\end{eqnarray*}
be the dispersion of the conditional log likelihood.
Then, we have
\begin{eqnarray*}
\ell(P^n_{XZ}, \varepsilon) = n H(X|Z) + \sqrt{n V(X|Z)} \Phi^{-1}(\varepsilon) + O(\log n),
\end{eqnarray*}
where $\Phi(\cdot)$ is the cumulative distribution function of the standard Gaussian
random variable.
\end{theorem}

\subsection{Exponential Bound}

In this section, we review the exponential bounds.

\begin{definition}
For $P_{XZ} \in {\cal P}({\cal X} \times {\cal Z})$, let
\begin{eqnarray*}
\phi(\rho|P_{XZ}) = \log \sum_z P_Z(z) \left( \sum_x P_{X|Z}(x|z)^{\frac{1}{1-\rho}}\right)^{1-\rho}.
\end{eqnarray*}
\end{definition}
We have the following.
\begin{theorem}[\cite{hayashi:10b}]
\label{theorem:large-deviation-bound}
For any $0 < \rho \le \frac{1}{2}$, we have
\begin{eqnarray*}
\mathbb{E}_F[ d(F|P_{XZ})] \le \frac{3}{2} |{\cal S}|^\rho e^{\phi(\rho|P_{XZ})}.
\end{eqnarray*}
\end{theorem}

\begin{definition}
For $\theta > 0$, $P_{XZ} \in {\cal P}({\cal X} \times {\cal Z})$, 
and $R_Z \in {\cal P}({\cal Z})$, let
\begin{eqnarray*}
H_{1+\theta}(P_{XZ}|R_Z) := - \frac{1}{\theta} \log \sum_{x,z} R_Z(z) \left( \frac{P_{XZ}(x,z)}{R_Z(z)} \right)^{1+\theta}
\end{eqnarray*}
be the conditional R\'enyi entropy of order $1+\theta$ relative to $R_Z$.
For $\theta = 0$, we define 
\begin{eqnarray*}
H_1(P_{XZ}|R_Z) &:=& \lim_{\theta \to 0} H_{1+\theta}(P_{XZ}|R_Z) \\
&=& H(X|Z) - D(P_Z \| R_Z).
\end{eqnarray*}
\end{definition}

By using Jensen's inequality and by setting $\rho = \frac{\theta}{1+\theta}$, we have
\begin{eqnarray*}
\phi(\rho|P_{XZ}) \le - \frac{\theta}{1+\theta} H_{1+\theta}(P_{XZ}|P_Z).
\end{eqnarray*}
Thus, we have the following slightly looser bound.
\begin{corollary}
\label{corollary:large-deviation-bound}
For $0 < \theta \le 1$, we have
\begin{eqnarray*}
\mathbb{E}_F[ d(F|P_{XZ})] \le \frac{3}{2} |{\cal S}|^{\frac{\theta}{1+\theta}} e^{- \frac{\theta}{1+\theta} H_{1+\theta}(P_{XZ}|P_Z)}.
\end{eqnarray*}
\end{corollary}

From Theorem \ref{theorem:large-deviation-bound} and 
Corollary \ref{corollary:large-deviation-bound}, we have the following.
\begin{theorem}
\label{theorem:key-length-large-deviation-bound}
We have
\begin{eqnarray}
\lefteqn{ \ell(P_{XZ}, \varepsilon)  } \nonumber \\
&\ge&  \sup_{0 < \rho \le \frac{1}{2}} \frac{- \phi(\rho|P_{XZ}) + \log (2 \varepsilon / 3) }{\rho} - 1 
\label{eq:large-deviation-bound-1} \\
&\ge& \sup_{0 < \theta \le 1} \frac{\theta H_{1+\theta}(P_{XZ}|P_Z) + (1+\theta) \log (2 \varepsilon /3) }{\theta} - 1.
\label{eq:large-deviation-bound-2}
\end{eqnarray}
\end{theorem}

\section{Hybrid Bound}
\label{section:hybrid}

In this section, we derive another bound from the leftover hash lemma
(Lemma \ref{lemma:left-over}). A basic idea is to use the smoothing in 
a similar manner as in the derivation of Theorem \ref{theorem:large-deviation-bound}.
However, we do not use the large deviation bound.

\begin{theorem}
\label{theorem:hybrid-bound-reference-system}
For any $0 < \eta \le \varepsilon$, we have
\begin{eqnarray}
\lefteqn{ \ell(P_{XZ},\varepsilon) } \nonumber \\
&\ge& \max_{0 \le \theta \le 1} \max_{R_Z} [ \theta H_{1+\theta}(P_{XZ}|R_Z)  \nonumber \\
&& ~~~~+ (1- \theta) H_\rom{s}^{\varepsilon - \eta}(P_{XZ}|R_Z) ] 
	+ \log 4 \eta^2 - 1.
	\label{eq:hybrid-bound}
\end{eqnarray}
\end{theorem}

\begin{proof}
We define the smoothed probability
\begin{eqnarray}
\label{eq:smoothed-probability}
\bar{P}_{XZ}(x,z) = P_{XZ}(x,z) \bol{1}\left[ - \log \frac{P_{XZ}(x,z)}{R_Z(z)} > r \right].
\end{eqnarray}
From Lemma \ref{lemma:left-over}, we have
\begin{eqnarray*}
\lefteqn{ \mathbb{E}_F\left[ d(F| \bar{P}_{XZ}) \right] } \\
&\le& \sqrt{ |{\cal S}| e^{- H_2(\bar{P}_{XZ}|R_z)}} \\
&=& \sqrt{ |{\cal S}| \sum_{x,z} \frac{\bar{P}_{XZ}(x,z)^2}{R_Z(z)} } \\
&\le& \sqrt{ |{\cal S}| \sum_{x,z} \frac{P_{XZ}(x,z)^{1+\theta}}{R_Z(z)^\theta} e^{- (1-\theta)r} } \\
&=& \sqrt{ |{\cal S}| \sum_{x,z} e^{- \theta H_{1+\theta}(P_{XZ}|R_Z) - (1- \theta) r}}.   
\end{eqnarray*}
By the triangular inequality, we have
\begin{eqnarray*}
\lefteqn{ \mathbb{E}_F\left[ d(F|\bar{P}_{XZ}) \right] } \\
&\le& 2 d(P_{XZ}, \bar{P}_{XZ} ) + \frac{1}{2} \sqrt{ |{\cal S}| e^{- \theta H_{1+\theta}(P_{XZ}|R_Z) - (1-\theta) r} } \\
&=& P_{XZ}\left\{ - \log \frac{P_{XZ}(x,z)}{R_Z(z)} \le r \right\} \\
&&~~~~~+  \frac{1}{2} \sqrt{ |{\cal S}| e^{- \theta H_{1+\theta}(P_{XZ}|R_Z) - (1-\theta) r} }.
\end{eqnarray*}
Thus, by setting $r = H_\rom{s}^{\varepsilon - \eta}(P_{XZ}|R_Z)$ and by taking $|{\cal S}|$ so that
\begin{eqnarray*}
\frac{1}{2} \sqrt{ |{\cal S}| e^{- \theta H_{1+\theta}(P_{XZ}|R_Z) - (1-\theta) r} } \le \eta,
\end{eqnarray*}
we have the statement of the theorem.
\end{proof}

Note that the bound in Theorem \ref{theorem:hybrid-bound-reference-system}
interpolates the lower bound in Theorem \ref{theorem:spectrum} and 
the bound in Eq.~(\ref{eq:large-deviation-bound-2}) of Theorem \ref{theorem:key-length-large-deviation-bound}.
More specifically, when the supremum in Eq.~(\ref{eq:hybrid-bound}) is achieved
by $\theta = 0$, then the bound in Eq.~(\ref{eq:hybrid-bound}) reduces to 
the bound in Theorem \ref{theorem:spectrum}.
To derive the bound in Eq.~(\ref{eq:large-deviation-bound-2}), we need some
large deviation calculation. By using Markov's inequality, we have
\begin{eqnarray*}
\lefteqn{ P_{XZ}\left\{ - \log \frac{P_{XZ}(x,z)}{R_Z(z)} \le r \right\}  } \\
&=&  P_{XZ}\left\{ \theta \log \frac{P_{XZ}(x,z)}{R_Z(z)} \ge - \theta r \right\} \\
&\le& \exp\left\{ \theta r - \theta H_{1+\theta}(P_{XZ}|R_Z) \right\}.
\end{eqnarray*}
Thus, we have
\begin{eqnarray}
\label{eq:lower-bound-on-spectrum}
H_\rom{s}^{\varepsilon - \eta}(P_{XZ}|R_Z) \ge H_{1+\theta}(P_{XZ}|R_Z) + \frac{1}{\theta} \log (\varepsilon - \eta).
\end{eqnarray}
By setting $\eta = \frac{\varepsilon}{3}$, $R_Z = P_Z$, and by substituting
Eq.~(\ref{eq:lower-bound-on-spectrum}) into Eq.~(\ref{eq:hybrid-bound}), we have
the bound in Eq.~(\ref{eq:large-deviation-bound-2}).

In \cite{hayashi:12d}, the optimal choice of $R_Z$ was shown to be
\begin{eqnarray}
R_Z(z) = \frac{\left( \sum_x P_{XZ}(x,z)^{1+\theta} \right)^{\frac{1}{1+\theta}}}{ \sum_z \left( \sum_x P_{XZ}(x,z)^{1+\theta} \right)^{\frac{1}{1+\theta}}}.
\label{eq:optimal-R}
\end{eqnarray}

\begin{remark}
To derive the bound in Eq.~(\ref{eq:large-deviation-bound-1}), we need to
use more complicated bounding.
To derive a non-asymptotic bound that subsume the bound in Eq.~(\ref{eq:large-deviation-bound-1}),
we need to introduce more artificial quantities instead of
$H_{1+\theta}(P_{XZ}|R_Z)$ and $H_\rom{s}^{\varepsilon - \eta}(P_{XZ}|R_Z)$.
\end{remark}

\section{Numerical Calculation}
\label{section:numerical-calcuation}

In this section, we consider the i.i.d. setting.
We consider the case such that $Z$ is obtained from $X$ throughout BSC, i.e.,
\begin{eqnarray}
P_{XZ}(x,x) = \frac{1 - q}{2}, 
~~~P_{XZ}(x, x+1) = \frac{q}{2}.
\label{eq:bsc-2}
\end{eqnarray}

In this case, since $P_Z$ is the uniform distribution on $\{0,1\}$, 
from Eq.~(\ref{eq:optimal-R}), the optimal choice of $R_Z$ is
$R_Z = P_Z$.
We have
\begin{eqnarray*}
\lefteqn{ P_{XZ}^n\left\{ - \log P_{X|Z}^n(x^n|z^n) \le r \right\} } \\
&=& B\left(n,q, \frac{r + n \log (1-q)}{\log \frac{1-q}{q}} \right),
\end{eqnarray*}
where $B(n,q,k)$ is the cumulative density function of the binomial trial.
Thus, the lower and upper bounds in Theorem \ref{theorem:spectrum} can be described as
\begin{eqnarray*}
\ell_{\rom{s},\rom{low}}(\varepsilon) \le \ell(P_{XZ}^n,\varepsilon) \le \ell_{\rom{s},\rom{up}}(\varepsilon),
\end{eqnarray*}
where
\begin{eqnarray}
\ell_{\rom{s},\rom{low}}(\varepsilon) 
&=& B^{-1}(n,q,\varepsilon - \eta) \times \log \frac{1-q}{q} \nonumber \\
&&~~~~~~~~~ - n \log (1-q) + \log 4 \eta^2 - 1
\label{eq:s-lower} \\
 \ell_{\rom{s},\rom{up}}(\varepsilon) 
&=&  B^{-1}(n,q, \varepsilon + \zeta) \times \log \frac{1-q}{q} \nonumber \\
&&~~~~~~~~~ - n \log (1-q) - \log \zeta 
\label{eq:s-upper}
\end{eqnarray}

For the distribution of the form in Eq.~(\ref{eq:bsc-2}), the bound
in Eqs.~(\ref{eq:large-deviation-bound-1}) and (\ref{eq:large-deviation-bound-2}) coincide.
We have
\begin{eqnarray*}
H_{1+\theta}(P_{XZ}|P_Z) 
= - \frac{1}{\theta} \log \left( q^{1+\theta} + (1-q)^{1+\theta} \right).
\end{eqnarray*}
Thus, the bounds in Theorem \ref{theorem:key-length-large-deviation-bound} can be described as
\begin{eqnarray*}
\ell(P_{XZ}^n,\varepsilon) \ge \ell_{\rom{e}, \rom{low}}(\varepsilon),
\end{eqnarray*}
where 
\begin{eqnarray}
\ell_{\rom{e}, \rom{low}}(\varepsilon)
&=& \sup_{0 < \theta \le 1} \frac{- n \log \left(q^{1+\theta} + (1-q)^{1+\theta} \right) }{\theta}  \nonumber \\
&&~~~~~~~+  \frac{(1-\theta)}{\theta} \log(2\varepsilon /3) - 1
\label{eq:e-lower}
\end{eqnarray}

Similarly, the bound in Theorem \ref{theorem:hybrid-bound-reference-system} can be described as
\begin{eqnarray*}
\ell(P_{XZ}^n,\varepsilon) \ge \ell_{\rom{h},\rom{low}}(\varepsilon),
\end{eqnarray*}
where 
\begin{eqnarray}
\lefteqn{ \ell_{\rom{h},\rom{low}}(\varepsilon) } \nonumber \\
&=& \max_{0 \le \theta \le 1}\left[ - n \log \left( q^{1+\theta} + (1-q)^{1+\theta} \right) + (1-\theta) \phantom{\frac{q}{q}} \right.  \nonumber \\ 
&& \left. \times \left\{ B^{-1}\left(n,q,\varepsilon - \eta \right) \times \log \frac{1-q}{q} - n \log (1-q) \right\} \right] 
 \nonumber \\
 && + \log 4 \eta^2 - 1.
\label{eq:h-lower}
\end{eqnarray}

For $\varepsilon = 10^{-10}$ and $q=0.11$, we plot $\ell_{\rom{s},\rom{low}}(\varepsilon)$, $\ell_{\rom{s},\rom{up}}(\varepsilon)$,
$\ell_{\rom{e}, \rom{low}}(\varepsilon)$, $\ell_{\rom{h},\rom{low}}(\varepsilon)$, and Gaussian approximation
derived by Theorem \ref{theorem:gaussian-approximation} in Fig.~\ref{Fig:fixed-epsilon}, where we set $\eta = \zeta = \frac{\varepsilon}{2}$.
From the figure, we can find that the exponential bound is better than the min-entropy bound up to
about $n = 10000$. The hybrid bound is better than both the exponential bound and the min-entropy bound.
The Gaussian approximation overestimate the lower bounds,
but it is sandwiched by the lower bounds and the upper bound.
\begin{figure}[t]
\centering
\includegraphics[width=\linewidth]{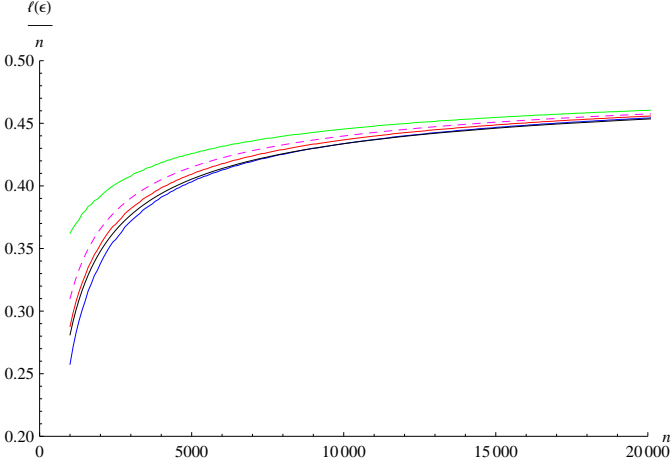}
\caption{A comparison among the bounds for $\varepsilon = 10^{-10}$ and $q=0.11$.
The blue curve is the min-entropy bound $\ell_{\rom{s},\rom{low}}(\varepsilon)$. 
The black curve is the exponential bound $\ell_{\rom{e}, \rom{low}}(\varepsilon)$.
The red curve is the hybrid bound $\ell_{\rom{h},\rom{low}}(\varepsilon)$.
The dashed pink curve is the Gaussian approximation.
The green curve is the upper bound $\ell_{\rom{s},\rom{up}}(\varepsilon)$.}
\label{Fig:fixed-epsilon}
\end{figure}

In Figs.~\ref{Fig:fixed-n-1000}, \ref{Fig:fixed-n-10000} and \ref{Fig:fixed-n-100000}, the bounds are compared 
from a different perspective, i.e., for fixed $n$ and varying $\varepsilon$.
From the figures, we can find that the exponential
bound and the hybrid bound become much better than the 
min-entropy bound as $\varepsilon$ becomes small.
When $\varepsilon$ is rather large for $n$, the min-entropy bound
is better than the exponential bound. 
The hybrid bound is better than both the exponential bound 
and the min-entropy bound for whole ranges of $\varepsilon$.

\begin{figure}[t]
\centering
\includegraphics[width=\linewidth]{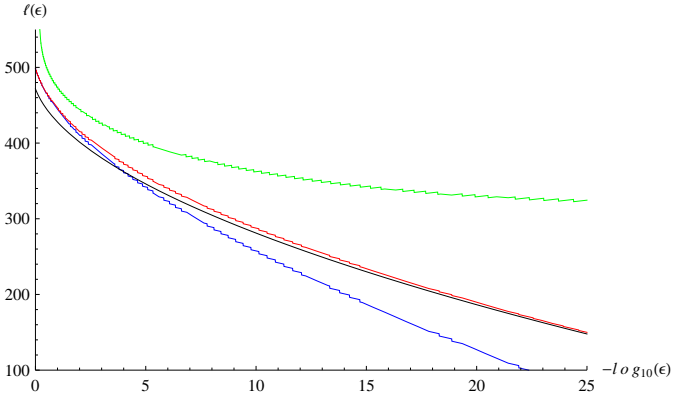}
\caption{A comparison among the bounds for $n = 1000$ and $q=0.11$.
The blue curve is the min-entropy bound $\ell_{\rom{s},\rom{low}}(\varepsilon)$. 
The black curve is the exponential bound $\ell_{\rom{e}, \rom{low}}(\varepsilon)$.
The red curve is the hybrid bound $\ell_{\rom{h},\rom{low}}(\varepsilon)$.
The green curve is the upper bound $\ell_{\rom{s},\rom{up}}(\varepsilon)$.}
\label{Fig:fixed-n-1000}
\end{figure}
\begin{figure}[t]
\centering
\includegraphics[width=\linewidth]{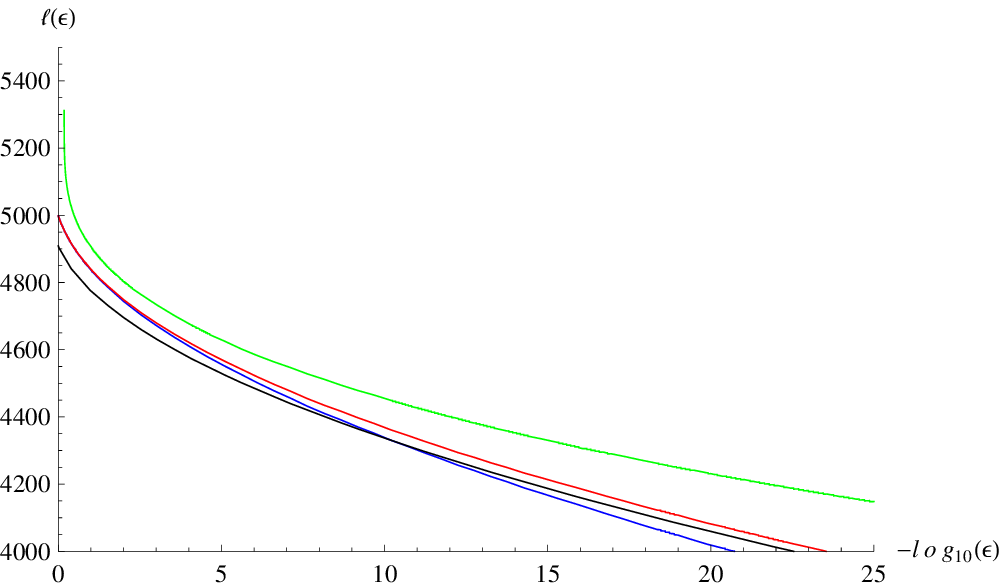}
\caption{A comparison among the bounds for $n = 10000$ and $q=0.11$.
The blue curve is the min-entropy bound $\ell_{\rom{s},\rom{low}}(\varepsilon)$. 
The black curve is the exponential bound $\ell_{\rom{e}, \rom{low}}(\varepsilon)$.
The red curve is the hybrid bound $\ell_{\rom{h},\rom{low}}(\varepsilon)$.
The green curve is the upper bound $\ell_{\rom{s},\rom{up}}(\varepsilon)$.}
\label{Fig:fixed-n-10000}
\end{figure}
\begin{figure}[t]
\centering
\includegraphics[width=\linewidth]{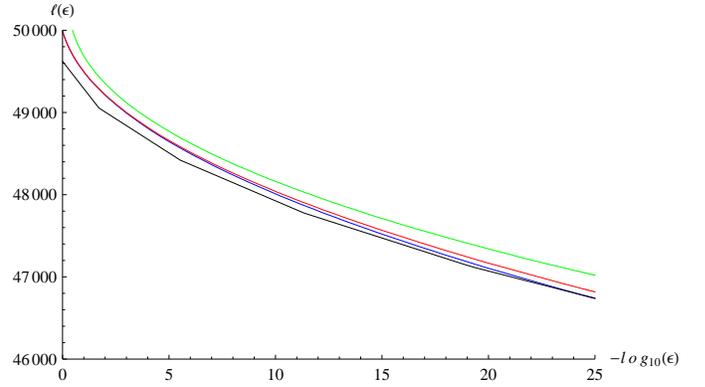}
\caption{A comparison among the bounds for $n = 100000$ and $q=0.11$.
The blue curve is the min-entropy bound $\ell_{\rom{s},\rom{low}}(\varepsilon)$. 
The black curve is the exponential bound $\ell_{\rom{e}, \rom{low}}(\varepsilon)$.
The red curve is the hybrid bound $\ell_{\rom{h},\rom{low}}(\varepsilon)$.
The green curve is the upper bound $\ell_{\rom{s},\rom{up}}(\varepsilon)$.}
\label{Fig:fixed-n-100000}
\end{figure}

\section{Conclusions}

In this paper, we have compared the exponential bound
and the min-entropy bound. It turned out that the exponential
bound is better than the min-entropy bound when $\varepsilon$
is rather small for $n$. When $\varepsilon$ is rather large for $n$,
the min-entropy bound is better than the exponential bound. We also 
presented the hybrid bound that interpolates the exponential
bound and the min-entropy bound.

For a future research agenda, it is important to extend the
results in this paper to the quantum setting or other information
theoretic security tasks such as the wire-tap channel.

\appendix

\subsection{Proof of Lemma \ref{lemma:monotonicity}}
\label{appendix:lemma:monotonicity}

Let $\tilde{P}_{SZ} \in {\cal B}^\varepsilon(P_{SZ})$ be such that 
\begin{eqnarray*}
H_{\min}^\varepsilon(P_{SZ}|R_Z) = H_{\min}(\tilde{P}_{SZ}|R_Z).
\end{eqnarray*}
Then, we define
\begin{eqnarray*}
\tilde{P}_{XZ}(x,z) = 
\tilde{P}_{SZ}(f(x),z) \frac{P_{XZ}(x,z)}{P_{SZ}(f(x),z)}.
\end{eqnarray*}
Then, we have
\begin{eqnarray*}
\lefteqn{ d(\tilde{P}_{XZ},  P_{XZ}) } \\
&=& \frac{1}{2} \sum_{x,z} |\tilde{P}_{XZ}(x,z) - P_{XZ}(x,z) | \\
&=& \frac{1}{2} \sum_{s,z} \sum_{x \in f^{-1}(s)} \frac{P_{XZ}(x,z)}{P_{SZ}(s,z)} |\tilde{P}_{SZ}(s,z) - P_{SZ}(s,z) | \\
&=& \frac{1}{2} \sum_{s,z}  |\tilde{P}_{SZ}(s,z) - P_{SZ}(s,z) | \\
&=& d(\tilde{P}_{SZ},  P_{SZ} ) \\
&\le& \varepsilon.
\end{eqnarray*}
Thus, we have $\tilde{P}_{XZ} \in  {\cal B}^\varepsilon(P_{XZ})$. Furthermore, by the construction of $\tilde{P}_{XZ}$, we have
$\tilde{P}_{XZ}(x,z) \le \tilde{P}_{SZ}(f(x),z)$ for every $(x,z)$. Thus, we have
\begin{eqnarray*}
H_{\min}^\varepsilon(P_{SZ}|R_Z) &=& H_{\min}(\tilde{P}_{SZ}|R_Z) \\
&\le& H_{\min}(\tilde{P}_{XZ}|R_Z) \\
&\le& H_{\min}^\varepsilon(P_{XZ}|R_Z).
\end{eqnarray*}
\qed

\subsection{Proof of Lemma \ref{lemma:spectrum-direct-bound}}
\label{proof:lemma:spectrum-direct-bound}

Let $r = H_\rom{s}^\varepsilon(P_{XZ}|R_Z)$. Then, let
\begin{eqnarray*}
\bar{P}_{XZ}(x,z) = P_{XZ}(x,z) \bol{1}\left[ - \log \frac{P_{XZ}(x,z)}{R_Z(z)} > r \right].
\end{eqnarray*}
Then, we have
\begin{eqnarray*}
d(P_{XZ}, \bar{P}_{XZ}) \le \frac{\varepsilon}{2}.
\end{eqnarray*}
Thus, we have $\bar{P}_{XZ} \in \bar{{\cal B}}^{\varepsilon/2}$. Thus, we have
\begin{eqnarray*}
\bar{H}_{\min}^{\varepsilon/2}(P_{XZ}|R_Z)
&\ge& H_{\min}(\bar{P}_{XZ}|R_Z) \\
&\ge& r \\
&=& H_\rom{s}^\varepsilon(P_{XZ}|R_Z).
\end{eqnarray*}
\qed

\subsection{Proof of Lemma \ref{lemma:spectrum-converse-bound}}
\label{proof:lemma:spectrum-converse-bound}

Let $r$ and $\tilde{P}_{XZ}$ be such that
$r = H_{\min}(\tilde{P}_{XZ}|P_Z) = H_{\min}^\varepsilon(P_{XZ}|P_Z)$.
For arbitrary fixed $\delta > 0$, let 
\begin{eqnarray*}
{\cal T} = \left\{ (x,z) : \frac{P_{XZ}(x,z)}{P_Z(z)} \le e^{- r + \delta} \right\}
\end{eqnarray*}
and
\begin{eqnarray*}
{\cal T}_z = \left\{ x : (x,z) \in {\cal T} \right\}.
\end{eqnarray*}
Then, we have $|{\cal T}_z^c| \le e^{r - \delta}$.
Furthermore, we have
\begin{eqnarray*}
d(P_{XZ}, \tilde{P}_{XZ}) 
&\ge& P_{XZ}({\cal T}^c) - \tilde{P}_{XZ}({\cal T}^c) \\
&=& P_{XZ}({\cal T}^c) - \sum_z \sum_{x \in {\cal T}_z^c} \tilde{P}_{XZ}(x,z) \\
&\ge& P_{XZ}({\cal T}^c) - \sum_z \sum_{x \in {\cal T}_z^c} e^{-r } P_Z(z) \\
&\ge& P_{XZ}({\cal T}^c) - e^{- \delta}.
\end{eqnarray*}
Thus, we have
\begin{eqnarray*}
\lefteqn{ r - \delta } \\ 
&\le& \sup\left\{ r^\prime : P_{XZ}\left\{ - \log \frac{P_{XZ}(x,z)}{P_Z(z)} \le r^\prime \right\} \le \varepsilon + e^{-\delta} \right\} \\
&=& H_\rom{s}^{\varepsilon + e^{-\delta}}(P_{XZ}|P_Z),
\end{eqnarray*}
which implies
\begin{eqnarray*}
H_{\min}^\varepsilon(P_{XZ}|P_Z)
&\le& H_\rom{s}^{\varepsilon + e^{-\delta}}(P_{XZ}|P_Z) + \delta \\
&=& H_\rom{s}^{\varepsilon + \zeta}(P_{XZ}|P_Z) - \log \zeta.
\end{eqnarray*}
\qed


\end{document}